\documentclass[11pt]{article}

\usepackage[margin=1in]{geometry}
\usepackage{fullpage,wrapfig,amsmath, amsthm, amssymb, algorithm,thmtools,algpseudocode,enumerate,caption,framed,paralist,sidecap,xspace}
\usepackage[usenames,dvipsnames,svgnames,table]{xcolor}
\definecolor{darkgreen}{rgb}{0.0,0,0.9}
\usepackage[colorlinks=true,pdfpagemode=UseNone,citecolor=Blue,linkcolor=Red,urlcolor=Black,pagebackref]{hyperref}
\usepackage{authblk}
\usepackage[colorinlistoftodos,prependcaption,textsize=tiny,]{todonotes}

\newtheorem{theorem}{Theorem}[section]
\newtheorem{lemma}{Lemma}[section]

\newcommand{\lleft}{\texttt{left}}
\newcommand{\Btop}{B_\texttt{top}}
\newcommand{\Bleft}{B_\texttt{left}}
\newcommand{\Bright}{B_\texttt{right}}
\newcommand{\Bbottom}{B_\texttt{bottom}}
\newcommand{\ttop}{\texttt{top}}

\title{(Faster) Multi-Sided Boundary Labelling\thanks{This work is supported in part by Natural Sciences and Engineering Research Council of Canada (NSERC).}}

\author[1]{Prosenjit Bose}
\author[1]{Saeed Mehrabi}
\author[2]{Debajyoti Mondal}

\affil[1]{{\small School of Computer Science, Carleton University, Ottawa, Canada.

\texttt{jit@scs.carleton.ca, saeed.mehrabi@carleton.ca}}
}

\affil[1]{{\small Department of Computer Science, University of Saskatchewan, Saskatoon, Canada.

\texttt{d.mondal@usask.ca}}
}

\date{}

\begin{document}

\maketitle

\begin{abstract}
A \emph{1-bend boundary labelling} problem consists of an axis-aligned rectangle $B$,  $n$ points (called \emph{sites}) in the interior, and $n$ points (called \emph{ports}) on the labels along the boundary of $B$. The goal is to find a set of $n$ axis-aligned curves (called \emph{leaders}), each having at most one bend and connecting one site to one port, such that the leaders are pairwise disjoint. A 1-bend boundary labelling problem is $k$-sided ($1\leq k\leq 4$) if the ports appear on $k$ different sides of $B$. Kindermann et al.~[``Multi-Sided Boundary Labeling'', Algorithmica, 76(1): 225-258, 2016] showed that the 1-bend three-sided and four-sided boundary labelling problems can be solved in $O(n^4)$ and $O(n^9)$ time, respectively. Bose et al.~[SWAT, 12:1-12:14, 2018] improved the latter running time to $O(n^6)$ by reducing the problem to computing maximum independent set in an outerstring graph. In this paper, we improve both previous results by giving new algorithms with running times $O(n^3\log n)$ and $O(n^5)$ to solve the 1-bend three-sided and four-sided boundary labelling problems, respectively.
\end{abstract}

%%%%%%%%%%%% NEW SECTION %%%%%%%%%%%%%%%%%
\section{Introduction}
\label{sec:introduction}
Map labelling is a well-known problem in cartography with applications in educational diagrams, system manuals and  scientific visualization. Traditional map labelling that places the labels on the map such that each label is incident to its corresponding feature, creates overlap between labels if the features are densely located on the map. This motivated the use of \emph{leaders}~\cite{DBLP:journals/ior/Zoraster97,FreemanMC96}: line segments that connect features to their labels. As a formal investigation of this approach, Bekos et al.~\cite{DBLP:journals/comgeo/BekosKSW07} introduced \emph{boundary labelling}: all the labels are required to be placed on the boundary of the map and to be connected to their features using leaders. The point where  a leader touches the label is called a \emph{port}. Their work initiated a line of research in developing labelling algorithms with different labelling aesthetics, such as minimizing leader crossings, number of bends per leader and sum of leader lengths~\cite{DBLP:journals/jgaa/BekosCFH0NRS15,DBLP:journals/algorithmica/BekosKNS10,DBLP:journals/jgaa/BenkertHKN09,BoseCK0M18}.

In this paper, we study the \emph{$1$-bend $k$-sided boundary labelling} problem. The input is an axis-aligned rectangle $B$ and a set of $n$ points (called \emph{sites}) in the interior of $B$. In addition, the input contains a set of $n$ points on the boundary of $B$ representing the ports on $k$ consecutive sides of $B$ for some $1\leq k\leq 4$. The objective is to decide whether each site can be connected to a unique port using an axis-aligned leader with at most $1$ bend such that the leaders are disjoint and each leader lies entirely in the interior of $B$, except the endpoint that is attached to a port. Figure~\ref{fig:partitionedExample}(a) illustrates a labelling for a 1-bend $k$-sided boundary labelling instance. If such a solution exists, then we call it a \emph{feasible} solution and say that the problem is \emph{solvable}. Notice that not every instance of the boundary labelling problem is solvable; see Figure~\ref{fig:partitionedExample}(b).

\begin{figure}[t]
\centering
\includegraphics[width=\textwidth]{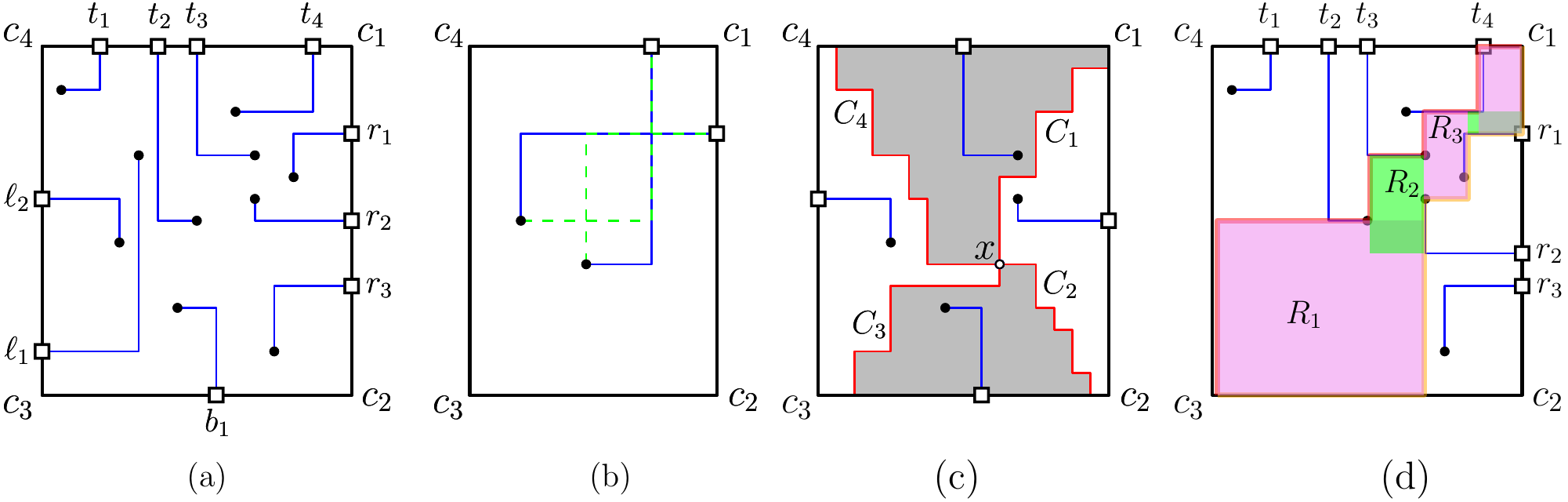}
\caption{(a) An 1-bend four-sided boundary labelling problem with $w=4, x=3,y=1$ and $z=2$, (b) an instance of the 1-bend two-sided boundary labelling problem with no planar solution with 1-bend leaders, and (c) a partitioned solution. (d) The sequence of empty rectangles.}
\label{fig:partitionedExample}
\end{figure}

\paragraph{Related work.} The boundary labelling problem was first formulated by Bekos et al.~\cite{DBLP:journals/comgeo/BekosKSW07} who solved the 1-bend one-sided and 1-bend two-sided models (when the ports lie on two opposite sides of $R$) in $O(n\log n)$ time. They also gave an $O(n\log n)$-time algorithm for the 2-bend four-sided boundary labelling problem (i.e., when each leader can have at most 2 bends). 

Kindermann et al.~\cite{DBLP:journals/algorithmica/KindermannNRS0W16} examined  $k$-sided boundary labelling, where the ports appear on adjacent sides. For the 1-bend two-sided boundary labelling problem, they gave an $O(n^2)$-time algorithm. For   1-bend three- and four-sided models of the problem, they gave  $O(n^4)$- and $O(n^9)$-time algorithms, respectively. If a boundary labelling instance admits an affirmative solution, then it is desirable to seek for a labelling that optimizes a labelling aesthetic, such as minimizing the sum of the leader lengths or minimizing the number of bends per leader. For minimizing the sum of leader lengths, Bekos et al.~\cite{DBLP:journals/comgeo/BekosKSW07} gave an exact $O(n^2)$-time algorithm for the 1-bend one-sided and 1-bend (opposite) two-sided models; their algorithm for 1-bend one-sided model was later improved to an $O(n\log n)$-time algorithm by Benkert et al.~\cite{DBLP:journals/jgaa/BenkertHKN09}.

Kindermann et al.~\cite{DBLP:journals/algorithmica/KindermannNRS0W16} gave an $O(n^8\log n)$-time dynamic programming algorithm   for the 1-bend (adjacent) two-sided model. Bose et al.~\cite{BoseCK0M18} (see~\cite{DBLP:journals/corr/abs-1803-10812} for the full version) improved this result by giving an $O(n^3\log n)$-time dynamic programming algorithm. They also showed that the 1-bend three- and four-sided problems (for the sum of the leader length minimization) can be reduced to the maximum independent set problem on outerstring graphs. The idea is to, for each site, obtain $n$ leaders by connecting the site to every port. Consider each leader as an outerstring (i.e., a polygonal curve that has one endpoint attached to the boundary of an enclosing rectangle).  This forms an outerstring graph. Then, by an appropriate weight assignment for edges, the boundary labelling problem is equivalent to the maximum-weight independent set problem on this outerstring graph. Since we have $O(n^2)$ outerstrings (where $n$ is the number of sites) and maximum-weight independent set can be solved in $O(N^3)$ time on outerstring graphs (here, $N$ is the complexity of geometrically representing the input graph)~\cite{KeilMPV17}, they obtain an $O(n^6)$-time algorithm for the three- and four-sided boundary labelling problem. However, it is not obvious whether this approach can be used to obtain a faster algorithm for the decision version, where we do not require leader length minimization.

Other models for boundary labelling such as labelling with sliding ports~\cite{DBLP:journals/jgaa/BenkertHKN09}, dynamic boundary labelling~\cite{DBLP:conf/gis/NollenburgPS10}, boundary labelling with octilinear leaders~\cite{DBLP:journals/algorithmica/BekosKNS10}, many-to-one boundary labelling~\cite{DBLP:journals/jgaa/LinKY08,DBLP:conf/apvis/Lin10} and boundary labelling in the presence of obstacles~\cite{DBLP:conf/cccg/0001S16} has also been studied. Moreover, see~\cite{DBLP:journals/jgaa/BenkertHKN09,BoseCK0M18} for results on bend minimization. Throughout the paper, we consider 1-bend leaders, which are also known as \emph{$po$-leaders}~\cite{DBLP:journals/comgeo/BekosKSW07}.

\paragraph{Our results.} In this paper, we give algorithms with running times $O(n^3\log n)$ and $O(n^5)$ for the 1-bend three-sided and four-sided boundary labelling problems, which improves the previously best known algorithms by nearly a linear factor.

The fastest known algorithm for the three-sided model was Kindermann et al.'s~\cite{DBLP:journals/algorithmica/KindermannNRS0W16} $O(n^4)$-time algorithm that reduced  the problem into $O(n^2)$ two-sided boundary labelling problems. While we also use their partitioning technique, our improvement comes from a dynamic programming approach that carefully decomposes three-sided problems into various simple shapes that are not necessarily rectangular. We prove that such a decomposition can be computed fast using suitable data-structures. The crux of our approach is to show new properties of simpler types of problems and then using them as  building blocks to solve the main labelling problem.

For the four-sided model, the fastest known algorithm was the $O(n^6)$-time algorithm of Bose et al.~\cite{BoseCK0M18} that reduced the problem into the maximum independent set problem in an outerstring graph. We show that Kindermann et  al.'s~\cite{DBLP:journals/algorithmica/KindermannNRS0W16} observation on partitioning  two-sided boundary labelling problems using a $xy$-monotone curve can be used to find a fast  solution for the four-sided model. Such an approach was previously taken by Bose et al.~\cite{BoseCK0M18}, but it already took $O(n^3\log n)$ time for the 2-sided  model, and they eventually settled with an $O(n^6)$-time algorithm for the four-sided model. With the four sides involved, designing a decomposition with a few different types of  shapes of low complexity becomes challenging. Our improvement results from a systematic decomposition  that generates a small number of subproblems at the expense of using a larger size dynamic programming table, resulting in an $O(n^5)$-time algorithm.

%%%%%%%%%%%%%%%%%%%%%%%%%%%%%%%%% NEW SECTION
\section{Preliminaries}
\label{sec:prelimins}
In this section, we give some notation and preliminaries that will be used in the rest of the paper. For a point $p$ in the plane, we denote the $x$- and $y$-coordinates of $p$ by $x(p)$ and $y(p)$, respectively. Consider the input rectangle $B$ and let $c_1,\dots,c_4$ denote the corners of $B$ that are named in clockwise order such that $c_1$ is the top-right corner of $B$. Let $\Btop, \Bbottom, \Bleft$ and $\Bright$ denote the top, bottom, left and right sides of $B$, respectively. We refer to a port as a \emph{top port} (resp., \emph{bottom, left} and \emph{right port}), if it lies on $\Btop$ (resp., $\Bbottom, \Bleft$ and $\Bright$). Similarly, we call a leader a \emph{top leader} (resp., \emph{bottom, left} and \emph{right leader}), if it is connected to a top port (resp., bottom, left and right port).

Let $w,x,y,z$ be the number of ports on the top, right, bottom and left side of $B$; notice that $w+x+y+z=n$. We denote these ports as $t_1,\ldots,t_w,$  $r_1,\ldots,r_x, b_1,\ldots,b_y$ and $\ell_1,\ldots,\ell_z$ in clockwise order. See Figure~\ref{fig:partitionedExample}(a) for an example. We assume that the sites are in general position; i.e., the number of sites and ports on every horizontal (similarly, vertical) line that properly intersects $B$ is at most one. For the rest of the paper, whenever we say a rectangle, we mean an axis-aligned rectangle.

For a point $x$ inside $B$, consider the rectangle $B_i$ that is spanned by $x$ and $c_i$, where $1\le i \le 4$. Each rectangle $B_i$ contains only two types of ports; e.g., $B_1$ contains only top and right ports. A feasible solution for a solvable instance of a boundary labelling problem is called \emph{partitioned}, if there exists a point $x$ such that for each rectangle $B_i$, there exists an axis-aligned $xy$-monotone polygonal curve $C_i$ from $x$ to $c_i$ that separates the two types of leaders in $B_i$. That is, every pair of sites in $B_i$ that lie on different sides of $C_i$ are connected to ports that lie on different (but  adjacent) sides of $B$. See Figure~\ref{fig:partitionedExample}(c). We refer to the polygonal curve as the \emph{$xy$-separating curve}. Kindermann et al.~\cite{DBLP:journals/algorithmica/KindermannNRS0W16} observed that if an instance of a boundary labelling problem admits a feasible solution, then it must admit a partitioned solution. 
\begin{lemma}[Kindermann et al.~\cite{DBLP:journals/algorithmica/KindermannNRS0W16}]
\label{lem:partitionedSolution}
If there exists a feasible solution for the 1-bend four-sided boundary labelling problem, then there also exists a partitioned solution for the problem.
\end{lemma}

Consider a two-sided problem where the ports are on $\Btop$ and $\Bright$. Assume that the problem has a feasible solution, and let $C$ be an $xy$-separating curve. Let $above(C)$ (resp., $below(C)$) be the polygonal regions above $C$ (resp., below $C$) that is bounded by $\Btop$ and $\Bleft$ (resp., by $\Bright$ and $\Bbottom$). Now, let $C_u$ (resp., $C_b$) be the $xy$-separating curve that minimizes the area of $above(C)$ (resp., $below(C)$). Given $C_u$ and $C_b$, we construct a sequence of rectangles as follows (see Figure~\ref{fig:partitionedExample}(d)).
\begin{itemize}
\item Each rectangle is a maximal rectangle between $C_u$ and $C_b$.
\item The bottom-left corner of $R_1$ is $c_3$. Since $R_1$ is maximal, it is uniquely determined.
\item Let $i>1$. We know that the top and right sides of $R_{i-1}$ are determined by a pair of leaders $L^t$ and $L^r$, respectively. Let $a\in L^t$ be the rightmost point on the top side of $R_{i-1}$, and let $b\in L^r$ be the topmost point on the right side of $R_{i-1}$. Then the rectangle $R_i$ is the maximal empty rectangle whose bottom-left corner is $(x(a),y(b))$ and that is bounded by $C_u$ and $C_b$.
\end{itemize}

We say that an instance of the problem is \emph{balanced} if it contains the same number of sites and ports. We next show the following result that will be useful in the following sections.
\begin{lemma}
\label{lem:onesided}
Let $P$ be a 1-bend one-sided boundary labelling problem with a constraint that the leftmost and  rightmost ports $a,c$ on $\Btop$ must be  connected to a pair of points $b,d$, respectively. Let $s$ be the rightmost or bottommost site of the problem excluding $b$ and $d$. If $P$ has a feasible solution satisfying the given constraint, then it also has a solution with connects $s$ to the first port $t_j$ (while walking from $c$ to $a$ along the boundary) that decomposes the problem into two balanced subproblems.
\end{lemma}
\begin{proof}
First, assume that $s$ is the bottommost point (see Figure~\ref{fig:onesided}(a)--(b)). Assume for a contradiction that connecting $s$ to $t_j$ would not give a feasible solution, whereas there is another port $t_i$ such that connecting $s$ to $t_i$ would yield a feasible solution. 

Note that $t_i$ lies to the left of  $t_j$. Let $L$ be the leader of $t_i$. We swap the leaders of $t_i$ and $t_j$.  Such a swap may introduce crossings in the rectangular region $R$ to the right side of   $L$. However, after the swap both sides of the leader of $t_j$ in $R$ are balanced  one-sided problems (see Figure~\ref{fig:onesided}(c)). Such a solution  with crossings to a 1-sided boundary labelling problem can always be made planar by local swaps~\cite{DBLP:journals/jgaa/BenkertHKN09}. The proof for the case when $s$ is the rightmost point is the same. Figure~\ref{fig:onesided}(e)--(h) illustrate such a scenario.
\end{proof}

\begin{figure}[t]
\centering
\includegraphics[width=\textwidth]{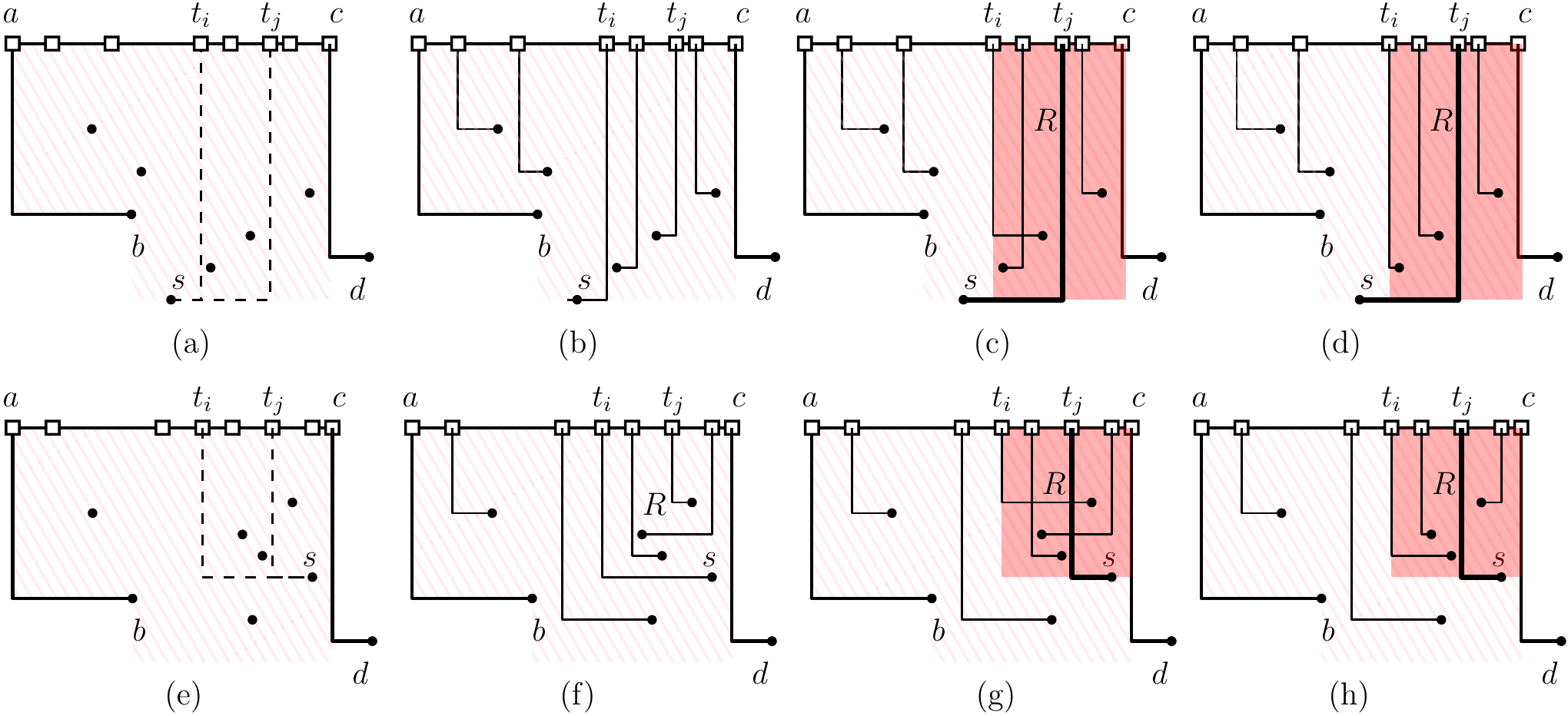}
\caption{An illustration in supporting the proof of Lemma~\ref{lem:onesided}.}
\label{fig:onesided}
\end{figure}

%%%%%%%%%%%%%%%%%%%%%%%%%%%%%%%%% NEW SECTION
\section{Three-Sided Boundary Labelling}
\label{sec:3sided}
In this section, we give an $O(n^3\log n)$-time algorithm for the three-sided boundary labelling problem. We assume that the ports are located on $\Bleft,\Btop$ and $\Bright$. Kindermann et al.~\cite{DBLP:journals/algorithmica/KindermannNRS0W16} gave an $O(n^4)$-time algorithm for this problem as follows. Consider the grid induced by a horizontal and a vertical line through every port and site. For each node of this grid, they partition the three-sided problem  into an $L$-shaped two-sided and a $\Gamma$-shaped two-sided problem (Figure~\ref{fig:3sidedProblemTypes}(a)--(b)). They showed that the three-sided problem is solvable if and only if there exists a grid node whose two two-sided problems both are solvable.

\subsection{Algorithm Overview}
We also start by considering every grid node $x$, but we employ a dynamic programming approach that expresses the resulting two-sided subproblems by $x$ and a port. We show that the two-sided subproblems can have $O(1)$ different types. This gives us $O(1)$ different tables, each  of size $O(n^3)$. We show that the running time to fill all the entries is $O(n^3\log n)$. For a grid node $x$, let $x'$ be the projection of $x$ onto $\Btop$. The point $x'$ splits the ports on $B$ into two sets: those to the left of $x'$ and those to the right of it. This means that the line segment $xx'$ can be extended to an axis-aligned curve with two bends that splits the problem into two balanced \emph{$L$-} and \emph{$\Gamma$-shaped} two-sided subproblems; see Figure~\ref{fig:3sidedProblemTypes}. Kindermann et al.~\cite{DBLP:journals/algorithmica/KindermannNRS0W16} observed that such a balanced partition is unique. Hence if there exists a balanced partition by extending $xx'$ to a 2-bend curve, then there must be a unique position where the number of sites to the left of the curve  matches the number of ports to the left of $x'$.

In the following, we describe the details of solving the $L$- and $\Gamma$-shaped two-sided problems. W.l.o.g., we assume that the two-sided problem contains the top-left corner $c_4$ of $B$.  While decomposing an $L$-shaped problem, we will reduce it either into  a smaller $L$-shaped problem or  a rectangular two-sided problem. Decomposition of $\Gamma$-shaped problem is more involved.

\begin{figure}[pt]
\centering
\includegraphics[width=\textwidth]{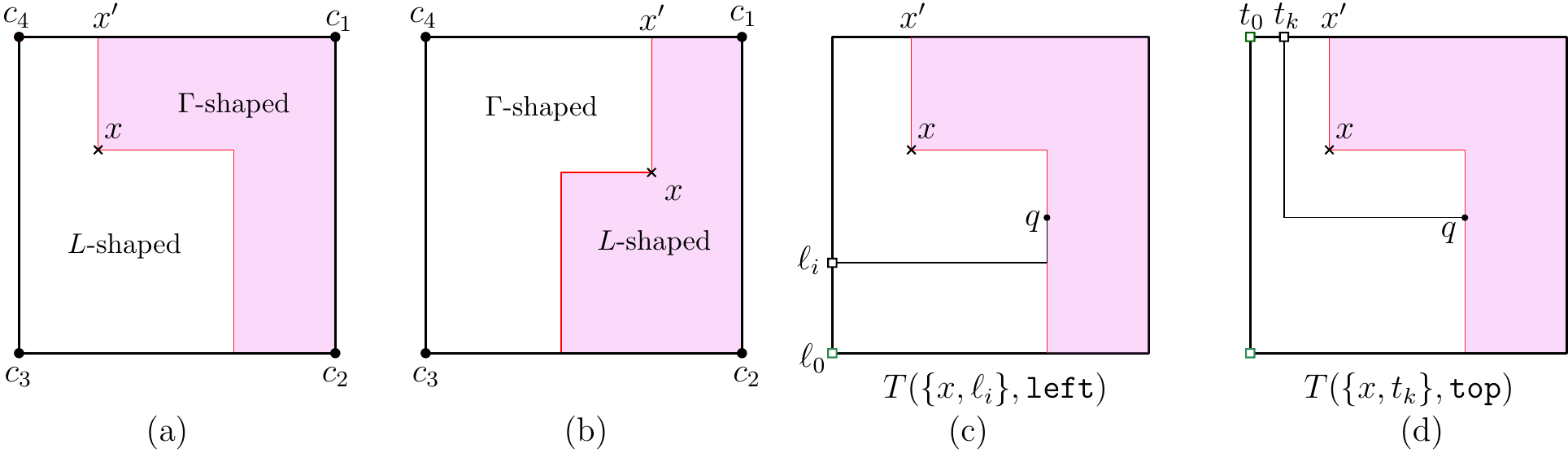}
\caption{(a)--(b) An $L$- and a  $\Gamma$-shaped  problem. (c)--(d)  Encoding of the $L$-shaped problems.  The  site $q$ that determines the boundary of the $L$-shape can be recovered from the encoding.}
\label{fig:3sidedProblemTypes}
\end{figure}

\subsection{Solving an $L$-Shaped Problem}
\label{subsec:3sidedTurnsLeft}
We denote an $L$-shaped subproblem with a grid node $x$ and a port $b$, where $b$ lies either on $\Btop$ or $\Bleft$. If $b$ is a port on $\Bleft$ (i.e., $b=\ell_i$ for some $i$), then we denote the subproblem by $T(\{x,\ell_i\},\lleft)$ in which the bottom side is determined by a horizontal line through $\ell_i$ (e.g., see Figure~\ref{fig:3sidedProblemTypes}(c)). Here, the last parameter denotes whether the port belongs to $\Bleft$ or $\Btop$. Initially, we assume a dummy port $\ell_0$ located at $c_3$ and so our goal is to compute $T(\{x,\ell_0\},\lleft)$. On the other hand, if $b=t_k$ for some $k$ (i.e., it is a port on $\Btop$), then we define the problem as $T(\{x,t_k\},\ttop)$ (see Figure~\ref{fig:3sidedProblemTypes}(d))). Here, the goal is to compute $T(\{x,t_0\},\ttop)$, where $t_0$ is a dummy port at $c_4$.

To decompose $T(\{x,\ell_i\},\lleft)$, we find  the rightmost site $p$  of the subproblem in $O(\log n)$ time (Figure~\ref{fig:3sidedDecompose}(a)). We first do some preprocessing, and then consider two cases depending on whether $p$ lies to the left or right side of the vertical line through $x$. 

For the preprocessing, we  keep the sites and ports in a range counting data structure that supports $O(\log n)$ counting query~\cite{Berg08}. Second, for each horizontal slab determined by a pair of horizontal grid lines $h,h'$ (passing through sites and ports), we keep the points inside the slab in a sorted array $M(h,h')$, which takes $O(n^3\log n)$ time.  We now use these data structures to find $p$. We first compute the sites and ports in the rectangle determined by the diagonal $c_4 x$, and then find the number of points needed in $O(\log n)$ time. Finally, we find a site $r$ (that balances the number of sites and points) in $O(\log n)$ time by a binary search in the array $M$ for the slab determined by the horizontal grid lines through $\ell_i$ and $x$. If $r$ lies to the right of the vertical line through  $x$, then $r$ is the desired point $p$. Otherwise, $p$ lies to the left of the vertical line through $x$, and we can find $p$ by searching in the array $M$ for the slab determined by the horizontal grid lines through $\ell_i$ and $t_0$.  We will frequently search for such a unique site throughout the paper and so we will use a similar preprocessing. 

\paragraph{Case 1 ($p$ lies to the right side of the vertical line through $x$).} We connect $p$ to the ``right'' port: after connecting, the resulting subproblems are balanced; i.e., the number of sites in each resulting subproblem is the same as the number of ports of that subproblem. There might be several ports that are right in this sense, but one can apply Lemma~\ref{lem:onesided} (assuming dummy leaders as boundary constraints) to show that the subproblem has a feasible solution if and only if there exists a feasible solution connecting $p$ either to the bottommost or to the rightmost  port that satisfies the balanced condition.   Once we find $p$ and the appropriate port $c$ for $p$, there are three possible scenarios as illustrated in Figure~\ref{fig:3sidedDecompose}. We can decompose the subproblem using the following recursive formula (depending on whether $p$ is connected to $\Bleft$ or $\Btop$):
\begin{equation*}
T(\{x,\ell_i\},\lleft) =   \begin{cases} T(\{x,\ell_j\},\lleft)  \wedge T' 
& \text{e.g., see  Figure~\ref{fig:3sidedDecompose}(a)--\ref{fig:3sidedDecompose}(b), or }\\
T(\{x,t_j\},\ttop) \wedge T(\{y,\ell_i\},\lleft) & \text{e.g., see Figure~\ref{fig:3sidedDecompose}(c)}
\end{cases}
\end{equation*}
Here, $T(\{x,t_j\},\ttop)$ is an $L$-shaped one-sided problem and $T'$ is a rectangular one-sided problem. There always exists a solution for the balanced rectangular one-sided problem~\cite{DBLP:journals/jgaa/BenkertHKN09}, and thus $T'$ can be considered as true. Since $T(\{x,t_j\},\ttop)$ is balanced, we can show that there always exists a solution for $T(\{x,t_j\},\ttop)$, as follows. First assume that the vertical line through $x$ does not pass through a port (see Figure~\ref{fig:3sidedTypeTwo}(a)). Let $t$ be the rightmost port in $T(\{x,t_j\},\ttop)$, and let $Q$ be the set of points inside the rectangle $R_x$ determined by diagonal $xq$. Scale down the rectangle  $R_x$ horizontally and translate  the rectangle inside the vertical slab determined by the vertical lines through $t$ and $x$.  There always exists a solution for the balanced rectangular one-sided problem~\cite{DBLP:journals/jgaa/BenkertHKN09}, we can translate the points back to their original position extending the leaders as necessary. The case when the vertical line through $x$  passes through a port $t_x$ can be processed in the same way, by first connecting $t_x$ to the topmost point of $R_x$, and then choosing the port immediately to the left of $t_x$ as $t$.  We thus have the following recurrence formula. 
\begin{equation*}
T(\{x,\ell_i\},\lleft)=\begin{cases} T(\{x,\ell_j\},\lleft)   & \text{ if $\ell_j$ exists, e.g., see Figure~\ref{fig:3sidedDecompose}(a)--\ref{fig:3sidedDecompose}(b), or}\\
T(\{y,\ell_i\},\lleft) & \text{ if $t_j$ exists, e.g., see  Figure~\ref{fig:3sidedDecompose}(c)}
\end{cases}
\end{equation*}
We now show how to find the point $p$ and then decompose  $T(\{x,\ell_i\},\lleft)$. The table $T(\{x,\ell_i\},\lleft)$ is of size $O(n^3)$. We will show that each subproblem  can be solved by a constant number of table look-ups, and these entries  can be found in $O(\log n)$ time. Hence, the overall computation takes $O(n^3 \log n)$ time.

\begin{figure}[t]
\centering
\includegraphics[width=\textwidth]{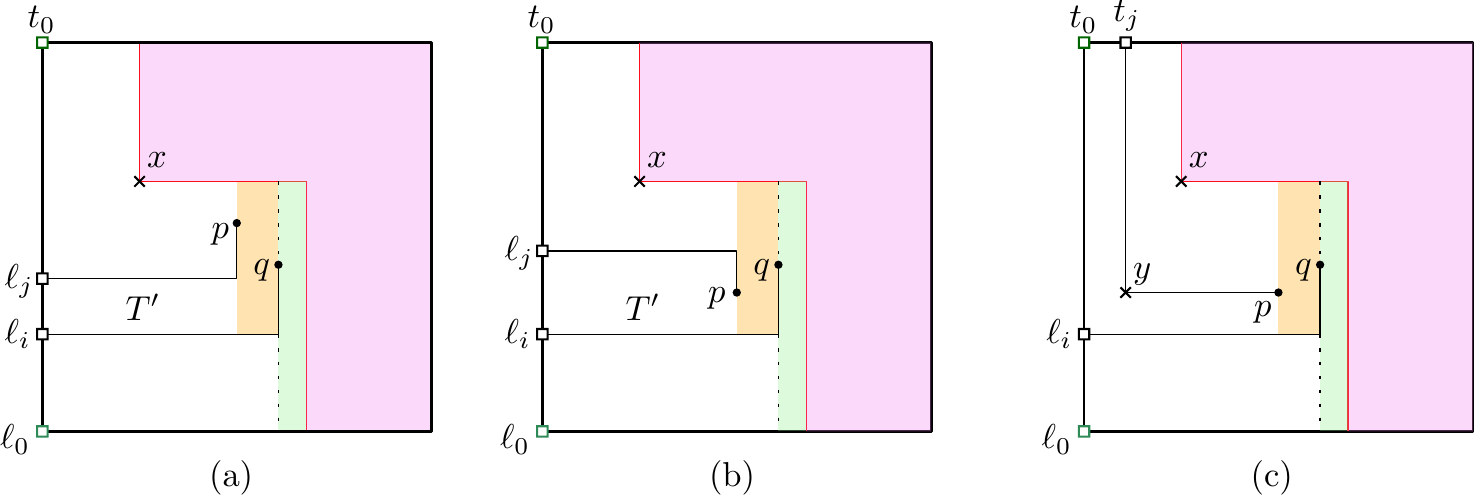}
\caption{The decomposition of an $L$-shaped problem when $b=\ell_i$.}
\label{fig:3sidedDecompose}
\end{figure}

\begin{figure}[h]
\centering
\includegraphics[width=\textwidth]{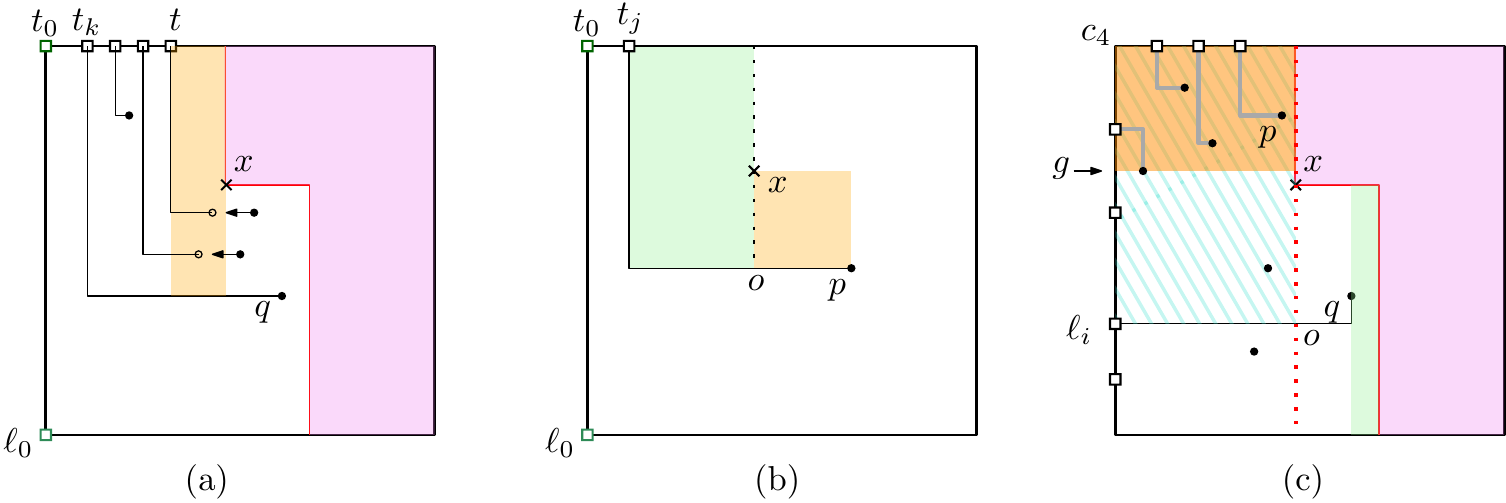}
\caption{(a) Reducing $T(\{x,t_j\},\ttop)$ to a rectangular one-sided problem, and (b) finding the port $c$ for $p$ when $c=t_j$ for some $j$. (c) Precomputation of the two-sided problems.}
\label{fig:3sidedTypeTwo}
\end{figure}

We now show how to find the ``right'' port  $c$ for $p$. First, assume that $c$ is a port $\ell_j$ on $\Bleft$ (see Figure~\ref{fig:3sidedDecompose}(a)--(b)). We need another $O(n^3\log n)$-time preprocessing as follows. Define a table $M'(s,\ell_i)$, where $s$ is a site, and $\ell_i$ is a port on $\Bleft$. At the entry $M'(s,\ell_i)$, we store  the port $\ell_j$ such that $j>i$ is the smallest index for which the rectangle defined by $\Bleft$, the horizontal lines through $\ell_i,\ell_j$ and the vertical line through $s$ contains exactly $(j-i-1)$ sites. We set $j$ to 0 when no such port $\ell_j$ exists. The table $M'$ has size $O(n^2)$ and we can fill each entry of the table in $O(n\log n)$ time; hence, we can fill out the entire $M'$ in $O(n^3\log n)$ time. Observe that the port $c$ for $p$ is stored in $M'(p,\ell_i)$.

Consider now the case when $c$ is a port $t_j$ on $\Btop$ (see Figure~\ref{fig:3sidedDecompose}(c)). We again rely on an $O(n^3 \log n)$-time preprocessing. For every site $s$ and a vertical grid line $\ell$ (passing through a port or a site) to the left of $s$, we keep a sorted array $M''(s,\ell)$ of size $O(n)$. Each element of the array corresponds to a rectangle bounded by the horizontal line through $s$,  $\Btop, \ell$, and another vertical grid line $\ell'$ through a port to the left of $\ell$. The rectangles are sorted based on the difference between the sites and ports and then by the $x$-coordinate of $\ell'$.  Consequently, to find $c(=t_j)$, we can look for the number of sites in the rectangle determined by the diagonal $px$ (see Figure~\ref{fig:3sidedTypeTwo}(b)), and then binary search for that number in the precomputed array for $M''(p,x)$.

\paragraph{Case 2 ($p$ lies to the left side of the vertical line through $x$).} Let $o$ be the intersection of the vertical line through $x$ and the leader connecting $\ell_i$ and $q$ (see Figure~\ref{fig:3sidedTypeTwo}(c)). Since $p$ is the rightmost point, it suffices to solve the rectangular two-sided problem $P$ determined by the rectangle with diagonal $oc_4$ (shown in falling pattern). We will determine whether a solution exists in $O(1)$ time based on some precomputed information, as follows.

For each vertical grid line, we will precompute the  topmost horizontal grid line $g$ such that the two-sided problem (shown in orange) determined by these lines, $\Btop$ and $\Bleft$ has a feasible solution. If $g$ lies above $\ell_i$, then the two-sided problem is  solvable (as the remaining region determines a balanced one-sided problem); otherwise, it is not solvable. We will use the known  $O(n^2)$-time algorithm~\cite{DBLP:journals/algorithmica/KindermannNRS0W16} to check the feasibility of a two-sided problem, which looks for a ``partitioned'' solution. However, we do not necessarily require our  two-sided problem to be partitioned.

We now show that the precomputation takes $O(n^3 \log n)$ time. For every vertical grid line $v$, we first compute a sorted array $A_v$ of horizontal grid lines such that the two-sided problem determined by  $v$ and each element of $A_v$ is balanced. This takes $O(n \log n)$ time for $v$ and $O(n^2 \log n)$ time for all vertical grid lines. For each $v$, we then do a binary search on $A_v$ to find the topmost grid line $g$ such that the corresponding two-sided problem has a feasible solution. Since computing a solution to the two-sided problem takes $O(n^2)$ time~\cite{DBLP:journals/algorithmica/KindermannNRS0W16}, $g$ can be found in $O(n^2 \log n)$ time for $v$, and in $O(n^3 \log n)$ time for all the vertical grid lines.

\paragraph{Remark.} While decomposing $L$-shapes, we always find the rightmost point $p$. If the $x$-coordinate of $p$ is larger than that of $x$, then the subproblems $T(\{x,\ell_j\},\lleft)$ and $T(\{y,\ell_i\},\lleft)$ are also $L$-shaped. Otherwise, the problem reduces to a rectangular one-sided or two-sided problem. Hence a $\Gamma$-shape problem does not appear during the decomposition of $L$-shaped problems.

\subsection{Solving a $\Gamma$-Shaped Problem}
\label{subsec:3sidedTurnsRight}
Consider the $\Gamma$-shaped problem containing the top-left corner $c_4$; see e.g. Figure~\ref{fig:turnsRight}(a). Let $o$ be the projection of $x$ onto $\Bbottom$, and let $y$ be the other bend of the $\Gamma$ shape. In the following, we refer to the rectangle with diagonal $oy$ as the \emph{forbidden region}.

Kindermann et al.~\cite[Lemma 8]{DBLP:journals/algorithmica/KindermannNRS0W16} observed that  there must exist an axis-aligned $xy$-monotone curve $C$ that connects $c_4$ to $x$ such that the sites above $C$ are connected to top ports and the sites below $C$ are connected to left ports. We extend $C$ to $o$ along the vertical line $xo$ (e.g., see   Figure~\ref{fig:turnsRight}(b)). Since no leader will enter the forbidden region, we can now consider $C$ as a separating curve  for a two-sided problem determined by the rectangle with diagonal $c_4o$. For any partitioned solution, we can compute a sequence of maximal empty rectangles, as we discussed in Section~\ref{sec:prelimins}, such that the lower-right corner of $R_1$ coincides with $o$. To decompose the problem, consider the first rectangle $R_1$ in this sequence. The rectangle $R_1$ can either cover the forbidden region entirely or partially in three different ways; see Figure~\ref{fig:turnsRight}(b). Since the separating curve through the grid point $x$ (Figure~\ref{fig:turnsRight}(a)) determines a partition, the leaders of the $\Gamma$-shape must not enter into the forbidden region. Therefore, it suffices to consider an empty rectangle $R_1$ that entirely covers the forbidden region.

\begin{figure}[pt]
\includegraphics[width=\textwidth]{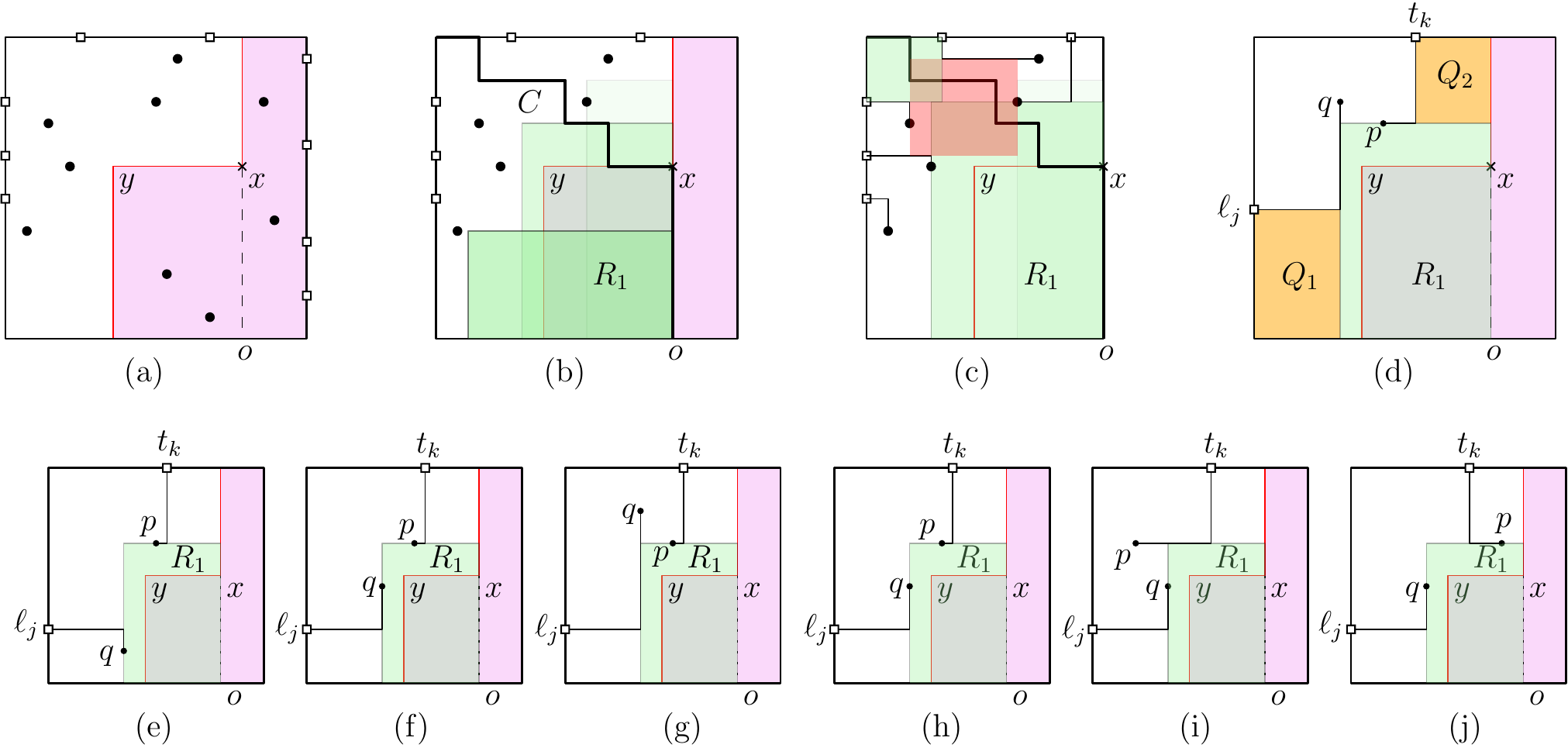}
\caption{(a) A $\Gamma$-shaped problem. (b)--(c) Three possibilities for the empty rectangle $R_1$, and a sequence of maximal rectangles corresponding to the separating curve. (d) Rectangle $R_1$ decomposes the problem into two one-sided subproblems and a smaller two-sided subproblem. (e)--(j) All possible cases for the two leaders determining the top and left sides of $R_1$.}
\label{fig:turnsRight}
\end{figure}

Since $R_1$ is maximal,  the top or left side of $R_1$ must contain a site (e.g., see   Figure~\ref{fig:turnsRight}(c)); because otherwise, the leaders that determines the top and left side of $R_1$ will cross. Assume w.l.o.g. that the top side of $R_1$ contains a site $p$ that is connected to a port $b(=t_k)$. Notice that $b$ cannot be a left port because it contradicts the existence of the curve $C$. Moreover, the left side of $R_1$ either contains a site or it is aligned with a leader that connects a port $\ell_j$ to a site $q$. The bottom row in Figure~\ref{fig:turnsRight} shows all possible scenarios; notice that $q$ cannot be connected to a top port, again because of having $C$. This decomposes the problem into three subproblems (see Figure~\ref{fig:turnsRight}(d)). Two one-sided subproblems $Q_1$ and $Q_2$, and a smaller two-sided subproblem. $Q_1$ is bounded by the leader connecting $\ell_j$ to $q$, the left side of $R_1$, $\Bbottom$ and $\Bleft$. The subproblem $Q_2$ is bounded by the leader connecting $p$ to $t_k$, $\Btop$, the curve determined by $x$, and the top side of $R_1$.

The smaller two-sided subproblem  is bounded by $\Bleft, \Btop$, the leaders incident to $t_k$ and $\ell_j$, and the boundary of $R_1$. We solve such two-sided problems by finding a sequence of maximal empty rectangles as described in Section~\ref{sec:prelimins}. The idea is inspired by Bose et al.'s~\cite{BoseCK0M18} approach to solve a two-sided  problem using a compact encoding for the table. Since we do not optimize the sum of leader lengths, our approach is much simpler.

\paragraph{Decomposing two-sided subproblems.} To decompose the two-sided subproblems, we use the following observation.
\begin{lemma}[Bose et al.~\cite{BoseCK0M18}]
\label{lem:linear}
If a two-sided boundary labelling problem with sites on $\Btop$ and $\Bleft$ has a feasible solution,  then there exists a partitioned solution where every maximal empty rectangle contains a site either on its top or on its left side. 
\end{lemma}

We first describe how to represent the two-sided problems. If the maximal empty rectangle $R$ contains a site on its top side (Figure~\ref{fig:turnsRight}(g)), then we represent the problem with the ports and the $x$-coordinate of the left side of $R$, i.e., $t_k,\ell_j$ and $x$-coordinate of $q$. Otherwise, we use the ports   and the $y$-coordinate of the top side of $R$.

Formally, define a table $T_y(\ell_i,t_j,z)$ (resp., $T_x(\ell_i,t_j,z)$), where $\ell_i$ and $t_j$ are two ports and $z$ denotes the $y$-coordinate (resp., $x$-coordinate) of the site to which the port $t_j$ (resp., $\ell_i$) is connected (see Figure~\ref{fig:3sidedTwoSided}). Given  $T_x(\ell_i,t_j,z)$ (similarly, $T_y(\ell_i,t_j,z)$), we can completely determine the problem boundary, as follows. Note that by definition  the left side of $R$ (the rectangle that defined the subproblem, shown in pink) has the $x$-coordinate $z$. Hence the site of $t_j$  lies on the top side of $R$ (by Lemma~\ref{lem:linear}). We can thus  perform a binary search to find this site. Once we have the two sites (and so the two incident leaders), we can find the lower right corner $q$ for the next candidate rectangle as defined in Section~\ref{sec:prelimins}.

\begin{figure}[h]
\includegraphics[width=\textwidth]{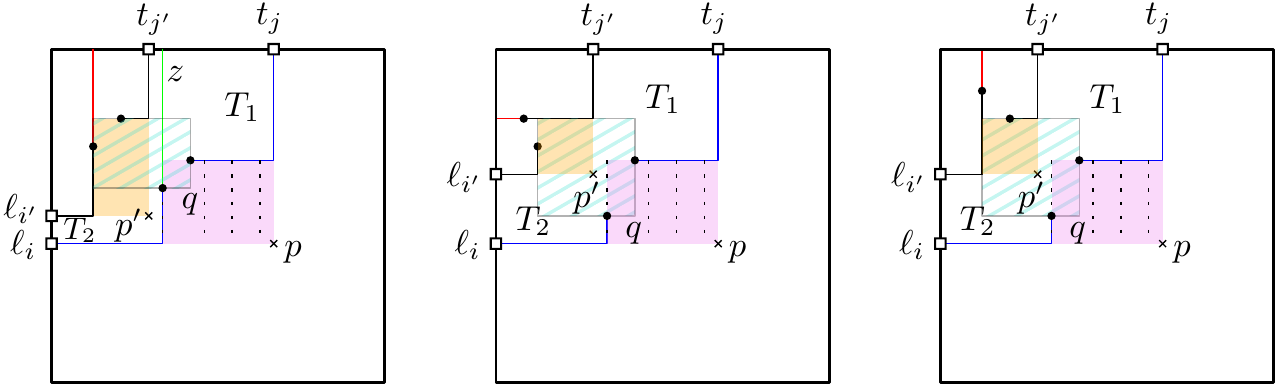}
\caption{An illustration for solving the two-sided subproblems arising in a $\Gamma$-shaped problem.}
\label{fig:3sidedTwoSided}
\end{figure}

We now show how to compute $T_x(\ell_i, t_j,z)$; the other is symmetric. After we determine the lower right corner $q$ for the next subsequent rectangle, we look for all candidate maximal empty rectangles $R'$ with lower right corner $q$. Consider now one such candidate  rectangle $R'$ (shown in orange in Figure~\ref{fig:3sidedTwoSided}), and let $\ell_{i'}$ and $t_{j'}$ be the ports defining its left and top sides, respectively. This now decomposes $T_x(\ell_i,t_j,z)$ into three subproblems; namely, two one-sided subproblems defined by $\ell_i,\ell_{i'}$ and $t_{j},t_{j'}$, and a  two-sided subproblem. By Lemma~\ref{lem:linear}, $R'$ contains a site on its top or  left side. If the top side of $R'$ contains a site, then we look-up the entry $T_x(\ell_{i'},t_{j'},z)$; otherwise, we look-up the entry $T_y(\ell_{i'},t_{j'},z)$. We hence have the following recurrence relation.
\begin{equation*}
T_x(\ell_i,t_j,z) = 
\begin{cases} T_x(\ell_{i'},t_{j'},z) \wedge T_1 \wedge T_2 & \text{if the top side of $R'$ contains a site,}\\
& \text{e.g., see Figure~\ref{fig:3sidedTwoSided}(left) and Figure~\ref{fig:3sidedTwoSided}(right)}\\
T_y(\ell_{i'},t_{j'},z) \wedge T_1 \wedge T_2 & \text{if the top side of $R'$ does not contain a site,}\\
& \text{e.g., see  Figure~\ref{fig:3sidedTwoSided}(middle)}.
\end{cases}
\end{equation*}

Here, $T_1$ and $T_2$ are one-sided problems, bounded by two leaders and a maximum empty rectangle. In the following, we discuss how to solve these problems.

\paragraph{Solving one-sided problems.} If $T_i$, where $i\in \{1,2\}$, is a balanced rectangular one-sided problem (Figure~\ref{fig:ones}(d)--(e)), then it has a feasible solution~\cite{DBLP:journals/jgaa/BenkertHKN09}. If $T_i$ is $L$-shaped and bounded by one side of the maximum empty rectangle (Figure~\ref{fig:ones}(f)), then it can be reduced to a rectangular one-sided problem, as described in Section~\ref{subsec:3sidedTurnsLeft}. Otherwise, the boundary of $T_i$ contains two sides of the maximal empty rectangle (Figure~\ref{fig:ones}(a)--(b)), or one side of the maximal empty rectangle and two vertical segments of the two leaders (Figure~\ref{fig:ones}(c)). In such a case, we can  decide whether it has a  feasible solution in $O(\log n)$ time, as follows.

Let $o,o'$ be the bottom-left corner of the empty rectangle $R'$ and the bend point of the leader of $\ell_i$ (Figure~\ref{fig:ones}(a)--(c)), respectively. Let $c$ be the number of sites in the open rectangle determined by the diagonal $oo'$. Let $\ell$ be a port such that there are $c$ ports between $\ell_i$ and $\ell$ (including $\ell$). If the $y$-coordinate of $\ell$ is not larger than  that of $o$ (or $a$, if $a$ is below $o$), then we can decompose the problem into two subproblems such that both of them have a feasible solution. One of these subproblems is a   rectangular one-sided problem  (shown in pink), and the other is an $L$-shaped one-sided problem (shown in green). Both of them are balanced. The existence of a solution is immediate from~\cite{DBLP:journals/jgaa/BenkertHKN09} for the rectangular case. The existence of a solution for the  $L$-shaped case is by the reduction to  a one-sided rectangular problem (Section~\ref{subsec:3sidedTurnsLeft}). If the $y$-coordinate of $\ell$ is larger than that of $o$ (or that of $a$, if $a$ is below $o$), then there does not exist a solution with leaders not intersecting the empty rectangle $R'$ (or the leader of $\ell_{i'}$).

\begin{figure}[h]
\centering
\includegraphics[width=\textwidth]{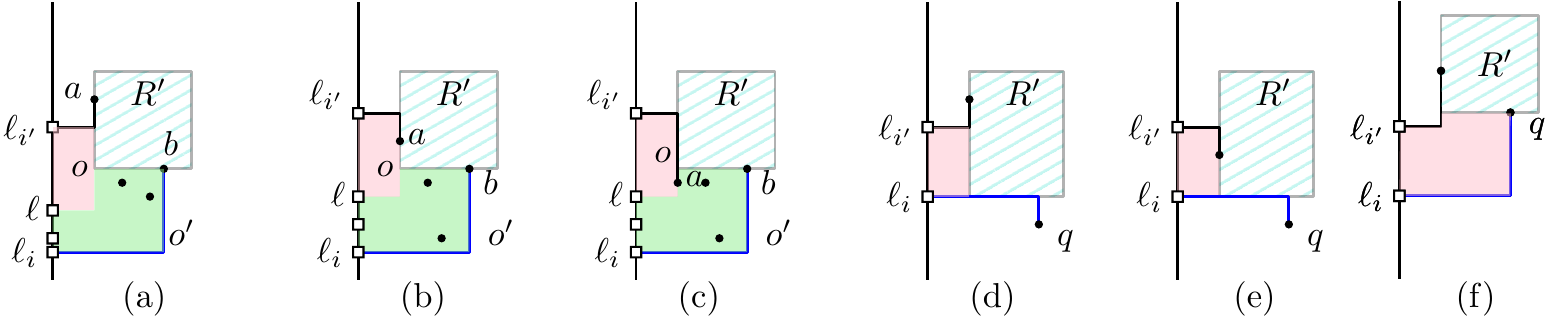}
\caption{An illustration for solving the one-sided subproblems arising in a $\Gamma$-shaped problem.}
\label{fig:ones}
\end{figure}

Note that the table $T_x$ and $T_y$ are of size $O(n^3)$ and each subproblem can be solved by examining a set of candidate rectangles. Bose et al.~\cite{BoseCK0M18} observed that the number of such candidate rectangles is $O(n^3)$, and while searching for candidate rectangles, no rectangle appears  more than once. Thus the total number of table look ups is bounded by $O(n^3)$. Furthermore, all these rectangles can be precomputed in $O(n^3 \log n)$ time~\cite{BoseCK0M18}. Hence, the overall time to solve a $\Gamma$-shaped problem is $O(n^3\log n)$. The following theorem summarizes the result of this section.
\begin{theorem}
\label{thm:3sided}
Given a 1-bend three-sided boundary labelling problem with $n$ sites and $n$ ports, one can find a feasible solution (if exists) in $O(n^3\log n)$ time.
\end{theorem}

%%%%%%%%%%%%%%%%%%%%%%%%%%%%%%%%% NEW SECTION
\section{Four-Sided Boundary Labelling}
\label{sec:4sided}
In this section, we give an $O(n^5)$-time dynamic programming algorithm for the four-sided boundary labelling.

\subsection{Algorithm Overview}
Our dynamic programming solution will search for a set of maximal empty rectangles $R_1, \ldots, R_k$ corresponding to the $xy$-separating curve, as described in Section~\ref{sec:prelimins}. The intuition is to represent a problem using at most two given leaders. For example, in Figure~\ref{fig:idea}(a) the first empty rectangle is $R_1$, with the top and right sides of $R_1$ determined by the leaders of $\ell_2$ and $b_2$, respectively. The problem will have a feasible solution with these given leaders if and only if the subproblems $P_1$ and $P_2$ (shown in falling and rising patterns, respectively) have feasible solutions. Since the subproblems must be balanced, given the two leaders adjacent to $P_2$ (see Figure~\ref{fig:idea}(a)), we can determine the boundary of the problem. We will use this idea to encode the subproblems.

It may initially appear that to precisely describe a subproblem, one should need some additional information along with the two given leaders. For example, the dashed boundary in Figure~\ref{fig:idea}(b). However, such information can be derived from the given leaders. Here, the top-right corner of $R_1$ (hence, the dashed line) can be recovered using the $y$-coordinate of the bottommost point of the leader  of  $\ell_2$, and the $x$-coordinate of the leftmost point of the leader of  $b_2$. We will try all possible choices for the first empty rectangle $R_1$. In a general step, we will continue searching for the subsequent empty rectangle. For example, consider the subproblem in Figure~\ref{fig:idea}(b). For a   subsequent empty rectangle, we will decompose the problem into at most three new subproblems (Figure~\ref{fig:idea}(c)). Each of these new subproblems can be represented using at most two leaders.

\begin{figure}[t]
\centering
\includegraphics[width=.9\textwidth]{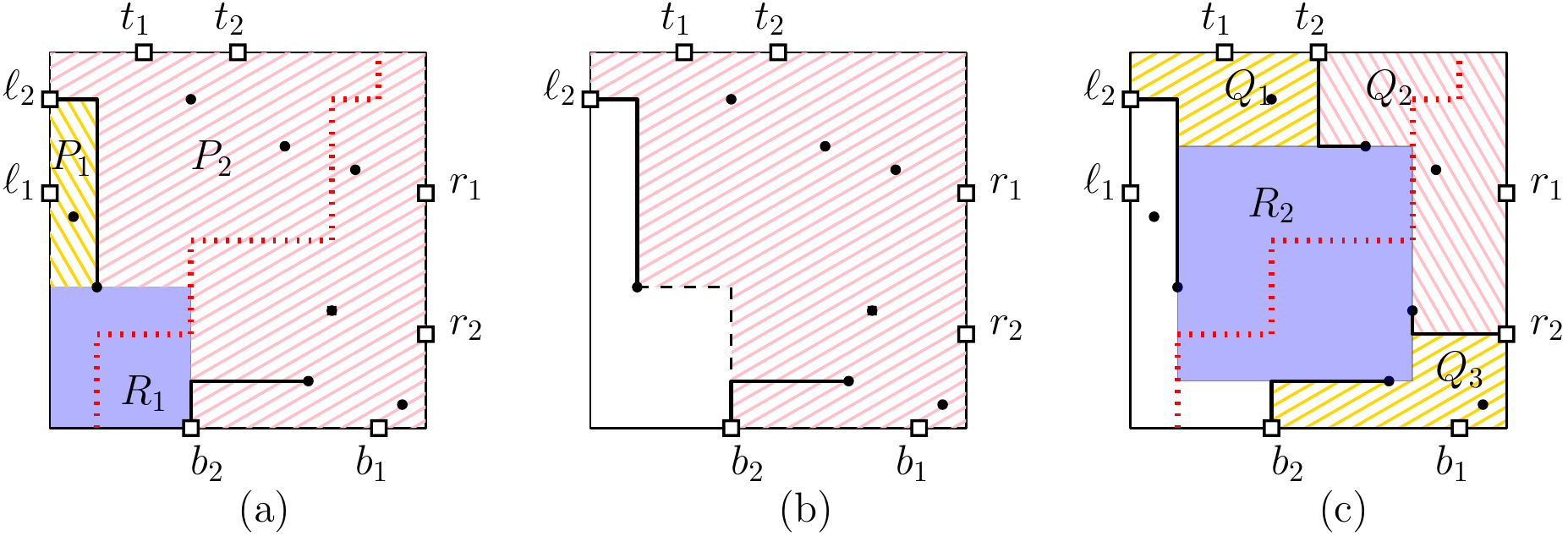}
\caption{The idea of decomposing a problem into subproblems.}
\label{fig:idea}
\end{figure}

\subsection{Solving Subproblems}
We will distinguish the subproblems depending on whether they contain the top-right corner $c_1$ of the rectangle $B$ or not. We first outline the case when a subproblem contains $c_1$.

\begin{figure}[t]
\includegraphics[width=\textwidth]{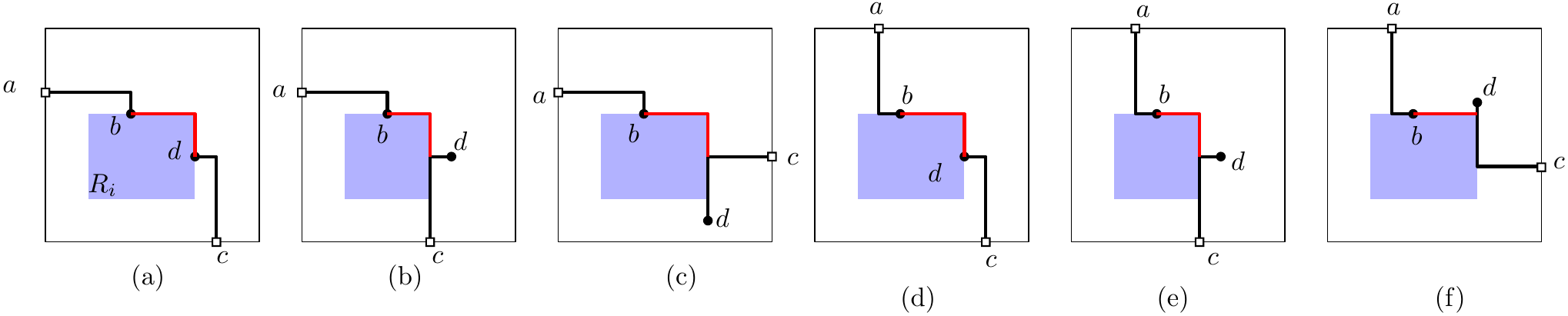}
\caption{Recovering the boundary (red) from the leaders when $R_i$ contains a site on its top side.}
\label{fig:simplecase}
\end{figure}

\paragraph{A subproblem contains $c_1$.} We denote a  subproblem that contains $c_1$ by $T(a,b,c,d)$. Here, $a,c$ are the two ports on the opposite sides of  the separating curve, and  $b$ and $d$ are their corresponding sites. Given $a,b,c,d$, we can completely determine the boundary of the subproblem, as follows.   Let $p$ and $q$ be a pair of  points on the leaders of $a$ and $c$, respectively,  with the minimum Euclidean distance. Then, the rectangle with diagonal $pq$ determines the boundary of the subproblem. Figure~\ref{fig:simplecase} illustrates the cases (up to the horizontal or vertical reflection of the sites), where the maximal empty rectangle $R_i$ contains a site on its top side. The recovered boundary is shown in red. On the other hand, Figure~\ref{fig:complexcase} illustrates the cases when $R_i$ does not contain any site on its top or right side.  If there are several choices for $p,q$, then we take the one with the smallest $y$-coordinate (or, symmetrically smallest $x$-coordinate) as illustrated in Figure~\ref{fig:complexcase}(d).

\begin{figure}[pt]
\centering
\includegraphics[width=.8\textwidth]{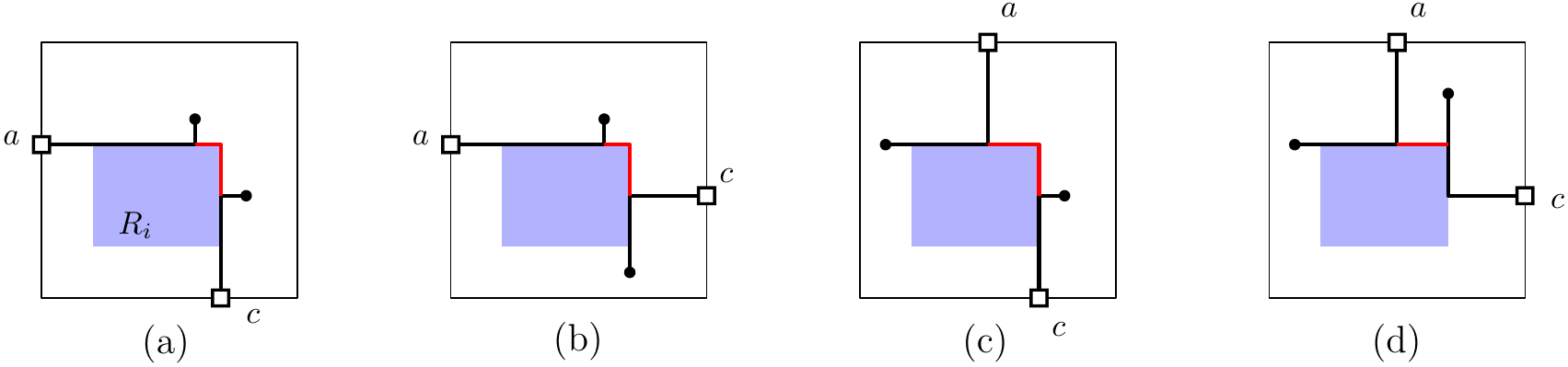}
\caption{Recovering the boundary (red) when $R_i$ does not contain  sites on its top or right side.}
\label{fig:complexcase}
\end{figure}

We now show how to decompose subproblems containing  $c_1$. Since we are searching for a partitioned solution, the decomposition must contain at least one such subproblem. Let $R_i$ be the empty rectangle associated with $T(a,b,c,d)$ (Figure~\ref{fig:decompose}(a)).  First assume the case when $c$ is on $\Bleft$ and $a$ is on $\Bbottom$. Let $p,q$ be the pair of points  with the minimum Euclidean distance between the  leaders of $c$ and $a$. Let $R_{pq}$ be the rectangle with diagonal $pq$.  Note that we can choose the bottom left corner of $R_{pq}$ as the bottom left corner of the   subsequent empty  rectangle $R_{i+1}$. There are $O(n^2)$ choices for the top-right corner of $R_{i+1}$. Thus trying all these subproblems would take $O(n^2)$ time to fill the table entry corresponding to $T(a,b,c,d)$.

\begin{figure}[pt]
\centering
\includegraphics[width=.9\textwidth]{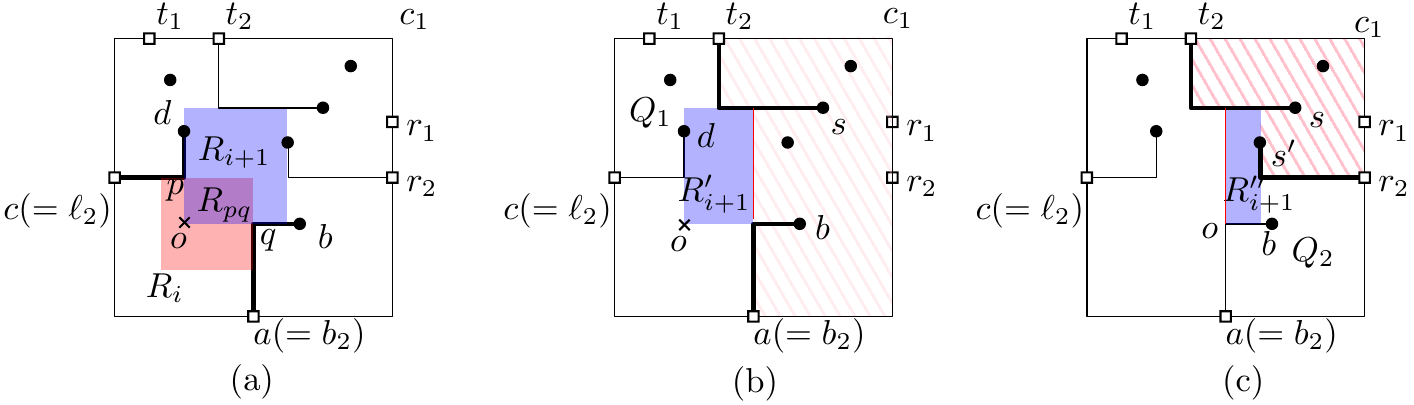}
\caption{Decomposition of a subproblem.}
\label{fig:decompose}
\end{figure}

Instead of searching for $R_{i+1}$ directly, we first extend $R_{pq}$ vertically (Figure~\ref{fig:decompose}(b)) and then horizontally (Figure~\ref{fig:decompose}(c)). Thus the subsequent rectangle $R'_{i+1}$ in our approach will have $O(n)$ choices for the top boundary. The leader $L$ that determines the top boundary will partition the problem into two subproblems: one is a problem that includes $c_1$, and the other problem $Q_1$ is either a one-sided or a two-sided subproblem depending on the position of $L$. For example, in Figure~\ref{fig:decompose}(b), $T(a,b,c,d)=T(a,b,t_2,s) \wedge Q_1$.

Note that sometimes $R_{pq}$ can be extended only in one direction. For example, in Figure~\ref{fig:decompose}(c), we have $T(a,b,t_2,s)$, and $R_{pq}$  is a vertical line bounded by the leaders, as shown in red. In this case we extend $R_{pq}$ horizontally and we have $O(n)$ choices for its right boundary. The leader $L'$  that determines the right boundary (i.e, the one connecting $r_2$ to $s'$) will partition the problem into two subproblems: one is a problem that includes $c_1$, and the other problem $Q_2$ is either a one-sided or a two-sided subproblem depending on the position of $L'$. For example, in Figure~\ref{fig:decompose}(c), $T(a,b,t_2,s)=T(r_2,s',t_2,s) \wedge Q_2$.

\paragraph{A subproblem excludes $c_1$.} We first show that these problems must be either one-sided or two-sided. Let $Q$ be the subproblem that includes $c_1$. Let $Q_1$ be a subproblem that  excludes $c_1$ and lies above the separating curve. Then the only available sides of $B$ for $Q_1$ are $\Bleft$ and $\Btop$. Symmetrically, only $\Bright$ and $\Bbottom$ are available for $Q_2$. If $Q_i$ (where $i\in\{1,2\}$) is a one-sided problem, then we can process the problem in the same ways as described in Section~\ref{sec:3sided}.

Consider now the case of two-sided problems, and assume w.l.o.g. that $Q_1$ is a two-sided problem. Note that $Q_1$ contains the top-left corner of $B$. Since we are searching for a partitioned solution (recall Lemma~\ref{lem:partitionedSolution}), there must be a $(-x,-y)$-monotone axis-aligned separating curve connecting the  bottom-right and the top-left  corners of $B$. Hence, there must exist a sequence of empty rectangles associated with  it.  Therefore, we can continue    decomposing the problem in the same way as we did for the problems that include $c_1$.

\subsection{Time Complexity}
We maintain the points and ports in an orthogonal range counting data structure (with $O(n\log n)$-time preprocessing) such that given an axis-aligned rectangle, one can report the number of ports and points interior to the rectangle in $O(\log n)$ time~\cite{Berg08}. 

The encoding of the form $T(a,b,c,d)$ allows us to use a dynamic programming table of size $O(n^4)$. Computing each entry requires examining $O(n)$ candidates for an empty rectangle. Each candidate rectangle divides the problem into at most two new subproblems. If the new subproblems are not balanced, then the candidate rectangle would not give a feasible solution. %\todo{A ``planar'' solution?}
 Otherwise, the new subproblems are balanced. For the two-sided problems   we perform table look-up, while the feasibility of the one-sided problems can be decided in $O(\log n)$ time. %to check whether each of them has an affirmative solution. 
 Once we complete the dynamic programming, we get a set of ports mapped to sites and a set of one-sided problems that are feasible, for which the solution needs to be constructed. Since there can be $O(n)$ such one-sided problems, and each can be solved in $O(n\log n)$ time (either by the algorithm of~\cite{DBLP:journals/jgaa/BenkertHKN09}, or by a reduction to the rectangular one-sided case), the total time required is $O(n^2 \log n)$.

We can use the range counting data structure to check whether a subproblem is balanced in $O(\log n)$ time~\cite{Berg08}. We can precompute the candidate rectangles for every grid point in $O(n^4 \log n)$ time. Since each table look-up takes $O(1)$ time, filling an entry of the dynamic programming table would take $O(n)$ time for all the $O(n)$ candidate rectangles corresponding to that entry. Hence, the overall running time of the algorithm is $O(n^5)$. The following theorem summarizes the main result of this section.
\begin{theorem}
\label{thm:mainResult}
Given a 1-bend four-sided boundary labelling problem with $n$ sites and $n$ ports, one can find a feasible labelling (if exists) in $O(n^5)$ time.
\end{theorem}

%%%%%%%%%%%%%%%%%%%%%%%%%%%%%%%%% NEW SECTION
\section{Conclusion}
\label{sec:conclusion}
In this paper, we gave algorithms with running times $O(n^3\log n)$ and $O(n^5)$ for the 1-bend three- and four-sided boundary labelling problems, improving the previous algorithms of Kindermann et al.~\cite{DBLP:journals/algorithmica/KindermannNRS0W16} and Bose et al.~\cite{BoseCK0M18}. The main direction for future work is to improve the running time of our  algorithms. Another direction is to study the fine-grained complexity of the problem; e.g., can we find a non-trivial lower bound on the running time? Even a quadratic lower bound is not known. Note that Bose et al.'s algorithm~\cite{BoseCK0M18} minimizes the sum of leader lengths, but ours do not. It would also be interesting to seek faster algorithms that minimize the sum of leader lengths.

\bibliographystyle{plain}
\bibliography{ref}

\end{document}